\documentclass[a4paper,UKenglish,cleveref, autoref, thm-restate]{lipics-v2021}

\usepackage[none]{hyphenat}
\usepackage{setspace}
\usepackage{latexsym,amsmath,amsthm}
\usepackage{amsfonts}
\usepackage{algorithm}
\usepackage[noend]{algpseudocode}
\usepackage{graphicx}
\usepackage[dvipsnames]{xcolor}
\usepackage{soul}
\usepackage{comment}
\usepackage{multibib}
\usepackage{caption}
\usepackage{subcaption}
\usepackage{xspace}

\newcites{supp}{Appendix References}

\DOIPrefix{}

\newcounter{claimcounter}
\newenvironment{smallclaim}[1]
{
  \refstepcounter{claimcounter}
    \noindent\textbf{Claim \theclaimcounter}: #1 \\
    \noindent\textit{Proof for Claim \theclaimcounter}:
}
{
  ~\qed
}

\includecomment{appendixonly}
\excludecomment{mainonly}

\newcounter{partcounter}
\newcommand{\myprefix}{}

\newcommand{\mylabel}[1]{\label{\myprefix#1}}
\newcommand{\myref}[1]{\ref{\myprefix#1}}
\newcommand{\myeqref}[1]{\eqref{\myprefix#1}}

\newcommand{\setprefix}[1]{\renewcommand{\myprefix}{#1}}

\makeatletter
\newcommand{\mypartprefix}{%
  \ifnum\value{partcounter}>0 % Check if partcounter is greater than 1
    \Alph{partcounter}%
  \fi
}

\newcommand{\mypart}{
  \stepcounter{partcounter}  % Increment the part counter
  \setcounter{section}{0}
  \setcounter{theorem}{0}
  \setcounter{lemma}{0}
  \setcounter{figure}{0}
  \setcounter{algorithm}{0}
}

\title{Optimal Delivery with a Faulty Drone}
\author{Jared {Coleman}}{University of Southern California, USA \and \url{https://jaredraycoleman.com}}{jaredcol@usc.edu}{https://orcid.org/0000-0003-1227-2962}{}
\author{Evangelos {Kranakis}}{Carleton University, Canada \and \url{https://people.scs.carleton.ca/~kranakis/}}{kranakis@scs.carleton.ca}{https://orcid.org/0000-0002-8959-4428}{}
\author{Danny {Krizanc}}{Wesleyan University, USA \and \url{https://dkrizanc.web.wesleyan.edu/}}{dkrizanc@wesleyan.edu}{}{}
\author{Oscar {Morales-Ponce}}{California State University, Long Beach, USA \and \url{https://www.csulb.edu/college-of-engineering/dr-oscar-morales-ponce}}{https://www.csulb.edu/college-of-engineering/dr-oscar-morales-ponce}{https://orcid.org/0000-0002-9645-1257}{}

\authorrunning{J. Coleman and E. Kranakis and D. Krizanc and O. Morales-Ponce} %TODO mandatory. First: Use abbreviated first/middle names. Second (only in severe cases): Use first author plus 'et al.'

\Copyright{Jared Coleman and Evangelos Kranakis and Danny Krizanc and Oscar Morales-Ponce} %TODO mandatory, please use full first names. LIPIcs license is "CC-BY";  http://creativecommons.org/licenses/by/3.0/

 %TODO mandatory; please add comma-separated list of keywords

\category{} %optional, e.g. invited paper

\relatedversion{} %optional, e.g. full version hosted on arXiv, HAL, or other repository/website
\nolinenumbers %uncomment to disable line numbering

\ArticleNo{1}
\begin{document}
\maketitle

\begin{abstract}
    We introduce and study a new cooperative delivery problem inspired by drone-assisted package delivery. We consider a scenario where a drone, en route to deliver a package to a destination (a point on the plane), unexpectedly loses communication with its central command station. The command station cannot know whether the drone's system has wholly malfunctioned or merely experienced a communications failure. Consequently, a second, helper drone must be deployed to retrieve the package to ensure successful delivery. The central question of this study is to find the optimal trajectory for this second drone. We demonstrate that the optimal solution relies heavily on the relative spatial positioning of the command station, the destination point, and the last known location of the disconnected drone. 
\end{abstract}

% Full Paper
\mypart
\setprefix{}
\includecomment{appendixonly}
\excludecomment{mainonly}
\section{Introduction}\mylabel{sec:intro}
In recent years, drone-based package delivery has emerged as a promising application of unmanned aerial vehicle (UAV) technology. As these systems are increasingly integrated into supply chain infrastructures, it becomes imperative to design algorithms for their robust operation amidst unexpected complications. In this paper, we consider a complication arising from faulty communication and propose a solution for a cooperative delivery problem inspired by this scenario.

Consider a situation where a drone, en route to deliver a package to a given destination, unexpectedly loses communication with its central command station. This unexpected loss of contact leaves the command station uncertain of whether the drone has suffered a communications breakdown or complete system failure. Furthermore, even if the issue \textit{is} only with the communications, the command station no longer has any way of knowing if/where the drone will fail on the rest of its way to the destination. In order to guarantee the package gets delivered, the command station must dispatch a second helper drone to retrieve the package and complete the delivery. Our goal is to design an online algorithm (one that cannot anticipate the true fail location of the drone) that, given the drone's last known location, determines the best trajectory for the second drone to find the package and complete the delivery in minimal time.

Formally, let us denote the last known location of the drone as the origin $S=(0, 0)$, the destination as the point $T=(1, 0)$, and the location of the command station as $P=(x, y)$, where $y \geq 0$ (all without loss of generality). The task is to identify the optimal trajectory for the second drone that minimizes the competitive ratio when compared to an optimal offline algorithm that knows the exact failure location $(t, 0)$ of the first drone in advance.
We assume the first drone will fail at some time $0 \leq t \leq 1$ (if it does not fail the delivery time is optimal and the problem is uninteresting).
Also, to simplify notation, we only consider the package to be delivered once the \textit{second} drone and the package are co-located at the destination (i.e., if $t=1$ and the first drone fails \textit{at} the destination, the package is only delivered once the second drone reaches $T$). 
For an algorithm $\mathcal{A}(x,y)$, which defines the trajectory of the second drone, let $A(x,y,t)$ denote the delivery time of algorithm $\mathcal{A}(x,y)$ for failure time $t$.
Our goal is to find an online algorithm (where $t$ is unknown) with minimum competitive ratio with respect to the delivery time of an optimal offline algorithm, Opt$(x,y,t)$ (where $t$ is known ahead of time).
The competitive ratio for algorithm $\mathcal{A}(x,y)$ for a given fail time $t$ can be written as
\begin{align*}
    \text{CR}_{\mathcal{A}(x,y)}(t) = \frac{A(x,y,t)}{\text{Opt}(x,y,t)} .
\end{align*}
Then the competitive ratio of $\mathcal{A}(x,y)$ is
\begin{align*}
    \text{CR}_{\mathcal{A}(x,y)} = \sup_t \text{CR}_{\mathcal{A}(x,y)}(t) .
\end{align*}
In order to simplify notation we sometimes eliminate $x$ and $y$ (when they are clear from context) and write algorithm $\mathcal{A}(x,y)$ as $\mathcal{A}$ and its delivery time $A(x,y,t)$ as $A(t)$.

\subsection{Model and Notation}\mylabel{sec:model}
In this section, we describe the model and notation used throughout the paper.
We call the first drone (the drone that is initially carrying the package towards the destination) the \textit{starter} and the second drone, which completes the delivery, the \textit{finisher}.
Without loss of generality, let $S = (0, 0)$ be the initial location of the starter (the point where it loses communication with the central command station) and $T = (1, 0)$ be the destination point on the plane. 
The starter drone begins at point $S$ carrying the package and moves following a straight line directly towards point $T$. 
At some unknown time $t \leq 1$, the starter will fail at position $(t, 0)$.

The finisher drone starts at a point $P=(x, y)$ on the plane (i.e., at the command station).
Without loss of generality, we assume that $y \geq 0$ (all results follow trivially by symmetry).
We assume that both drones have a speed of $1$ and always move at this speed.
The finisher can start, stop, and change direction instantaneously.
The drones can communicate with each other and exchange the package only when they are co-located (face-to-face communication). 
The package is considered to be delivered as soon as the finisher and the package are co-located at $T$.

Let $D(c, r)$ ($\overline{D}(c,r)$) denote the open (closed) disk with radius $r$ centered at point $(c,0)$. 
We use capital letters to denote points, $|PQ|$ to denote the Euclidean distance between two points $P$ and $Q$, and $\overline{PQ}$ to denote the line segment with endpoints $P$ and $Q$.

\subsection{Related Work}\mylabel{sec:related}

% Competitive analysis has been used to study many interesting search and (more recently) delivery problems with cooperative mobile agent systems.

%Misc Failures
An important aspect of our problem is that the starter agent experiences a failure. Many problems with cooperative mobile agent experiencing failures have been studied for a variety of basic problems in distributed computing and in various domains.
In \cite{faulty_search}, the authors study the search problem on the line by $n$ mobile agents where $f$ of the agents are faulty.
They present algorithms and their competitive ratios for different values of $f$. 
In \cite{DBLP:conf/algosensors/GeorgiouKLPP19} the authors study the evacuation problem on the disk by $n$ agents, $f$ of which may be faulty.  Competitive algorithms for gathering have been proposed for $n$ agents in a synchronous system with less than $(n-1)/3$ failures~\cite{faulty_gathering}. Optimal algorithms and hardness results for the multi-agent patrolling problem with faulty agents were presented in~\cite{faulty_patrolling}.
Algorithms for flocking~\cite{faulty_flocking} and evacuation~\cite{faulty_evacuation} have also been studied in the context of faulty agents.

%Search and Delivery
The problem studied in this paper has both search and delivery components. Many search problems have been studied for different domains and under different models and assumptions (cf. the book~\cite{search_gal_survey}).
The authors of \cite{del_pony_express} consider the delivery problem for messages on the line segment by multiple agents with different speeds and propose competitively optimal algorithms.
The results of the previous study were also extended to the plane~\cite{del_coleman_plane}.
Joint search and delivery problems, however, have received much less attention in the literature.
The problem has also been studied for the single- and two-agent case on the line~\cite{sd_coleman_line}. 
On the systems research side, many solutions for drone-assisted package delivery have been proposed for different environments and under different assumptions (we refer the reader to the survey~\cite{drone_survey}).

%Delivery
Most related but different to our delivery problem are the papers \cite{del_carvalho_graph} and \cite{del_drone_graph} on package delivery. In the former, the authors investigate delivery of one or two packages of many autonomous mobile agents initially located on distinct 
nodes of a weighted graph. In the latter paper, they are concerned with delivering a package from a source node to a destination node in a graph using a set of drones and study the setting where the movements of each drone are restricted to a certain subgraph of the given graph. Note that both papers above address the delivery problem in a graph setting. To the best of our knowledge our paper is the first to address delivery of a package in the plane in the presence of a faulty drone.

\subsection{Results}
The main result we present is an optimal online algorithm that depends only on the starting position of the finisher (the central command station).
Essentially, the algorithm executes one of three candidate algorithms depending on the finisher's starting position.
This is depicted in Figure~\myref{fig:regions_simple}.
\begin{figure}[!ht]
    \centering
    \includegraphics[width=0.5\linewidth]{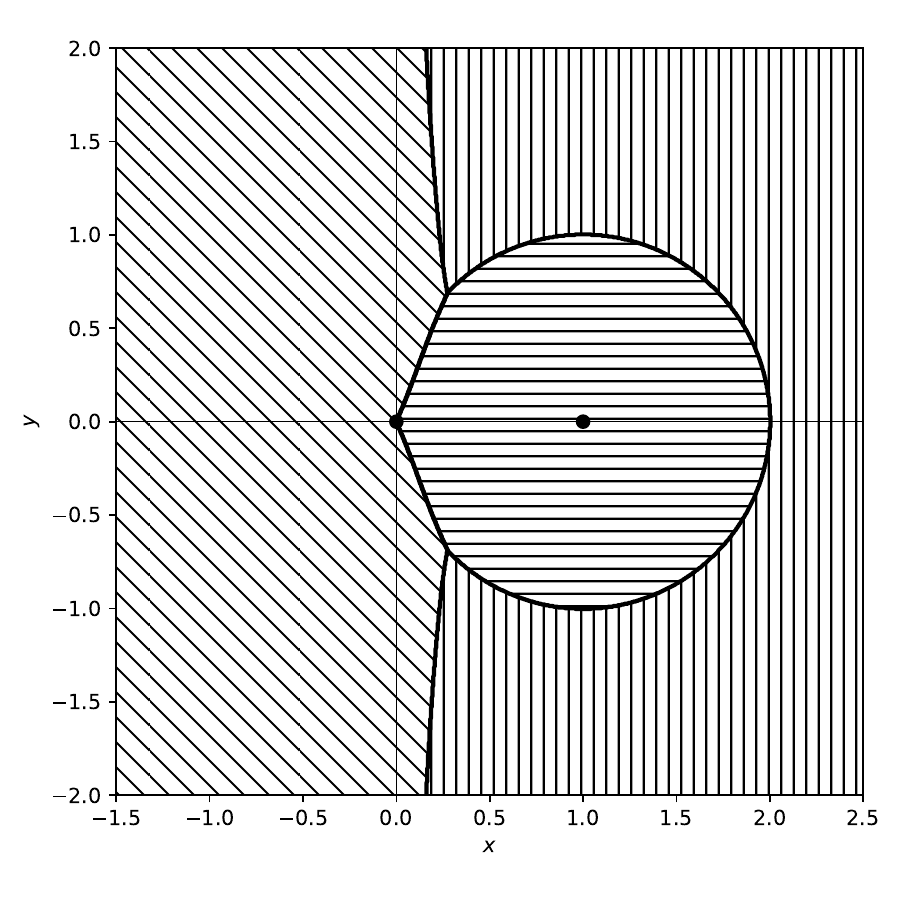}
    \caption{The optimal algorithm depends on the starting position of the finisher. The striped regions depict the finisher starting positions for which each of the three candidate algorithms is optimal.}
    \mylabel{fig:regions_simple}
\end{figure}
If the finisher starts in the diagonally striped region,
then the optimal algorithm is for the finisher to go to $S$ (the origin) and then towards the destination until finding the failed starter with the object.
If the finisher starts in the vertically striped region, the optimal algorithm is for the finisher to go first to $T$ (the destination) and then toward $S$ until finding the failed starter with the object.
Finally, if the finisher starts in the horizontally striped region, then the optimal algorithm is for the starter to go to the point $M = \left( \left(x^2+y^2\right)/(2x), 0 \right)$ and then toward $S$ until finding the failed starter with the object.
We will show that the finisher is guaranteed to find the starter on the interval $[0, M]$.

The rest of the paper is organized as follows.
In Section~\myref{sec:main}, we detail the three candidate algorithms and derive each of their competitive ratios as a function of the command station's starting position.
Then, in Section~\myref{sec:hybrid}, we present a hybrid algorithm that chooses the best of the three candidate algorithms to run based on the command station's starting position $(x,y)$.
We prove this algorithm is optimal, and then (in Section~\myref{sec:discussion}) discuss how the command station's position affects which of the candidate algorithms is executed in the optimal hybrid algorithm.
We conclude the paper in Section~\myref{sec:conclusion} with a summary of the results and a discussion of future directions.
\begin{mainonly}
All proofs omitted due to space constraints can be found in the full version of the paper in the appendix.
\end{mainonly}

\section{Candidate Algorithms}\mylabel{sec:main}
In this section, we present three algorithms and derive their competitive ratios.
In order to do so, we must first consider the optimal offline algorithm, where the fail location $(t, 0)$ (at time $t$) is known ahead of time and can be used to compute the optimal trajectory for the finisher.
Clearly in this case the finisher should go directly from its starting location to the starter's fail location $(t,0)$ and then complete the delivery.
The delivery time then can be written:
\begin{align*}
    \text{Opt}(t) = \max\left\{1, \sqrt{(x-t)^2+y^2} + 1 - t \right\} .
\end{align*}
Indeed, although the finisher moves directly to $(t, 0)$, it may have to wait for the starter to arrive (if necessary) prior to completing the delivery task.

For the online algorithms, we start by reasoning about what an optimal algorithm looks like.
First, since the starter must fail at some point on the line segment $\overline{ST}$, the finisher must \emph{eventually} move from its initial location to some point $(m, 0)$ on the segment (otherwise it will never find the starter with the package).
We use this to prove the following intuitively obvious lemma:

\begin{lemma}\mylabel{lm:alg_class}
    There exists an online algorithm with optimal competitive ratio that involves the finisher moving from its initial position $P=(x,y)$ directly to a point $M=(m, 0) \in \overline{ST}$, past which, it remains within the line segment $\overline{ST}$.
\end{lemma}
% \begin{appendixonly}
\begin{proof}
    %The proof is rather trivial.
    Let $M=(m,0)$ denote the first point the finisher reaches on $\overline{ST}$.
    We first show that the finisher should move \emph{directly} to $M$.
    Consider, for the sake of contradiction, an algorithm that involves the finisher moving from its initial position $P=(x,y)$ to a point $P_1 = (x_1,y_1) \notin \overline{PM}$ and then to point $M=(m,0)$ before finding the starter and delivering the package to the destination.
    Clearly, since the finisher could not have found the starter at any point not in $\overline{ST}$, and $|P P_1| + |P_1 M| > |PM|$ (triangle inequality), then the algorithm where the finisher moves directly to $M$ delivers the package at least as fast.

    Now we show that the finisher should stay on the line segment $\overline{ST}$ after reaching it for the first time.
    Again, for the sake of contradiction, consider an algorithm which involves the finisher moving from $M=(m,0) \in \overline{ST}$ to a point $P_2 = (x_2,y_2) \notin \overline{ST}$ and then back to a point $M_1=(m_1,0) \in \overline{ST}$ before finding the starter and delivering the package to the destination.
    Since the finisher could not have found the failed starter at any point not in $\overline{ST}$, and $|M P_2| + |P_2 M_1| > |M M_1|$ (triangle inequality), then the algorithm where the finisher moves directly from $M$ to $M_1$ (never leaving $\overline{ST}$) delivers the package at least as fast.
\end{proof}
% \end{appendixonly}

Now, we present three candidate online algorithms:
\begin{enumerate}
    \item \textbf{$\mathcal{A}_0$} (\textbf{Go To Last Point of Contact}): The finisher moves to the origin $S=(0,0)$ and then towards the destination until it finds the failed starter (trajectory $P \rightarrow S \rightarrow T$).
    \item \textbf{$\mathcal{A}_1$} (\textbf{Go To Destination}): The finisher goes first to the destination $T=(1,0)$ and then towards the origin until it finds the failed starter (trajectory $P \rightarrow T \rightarrow S \rightarrow T$).
    \item \textbf{$\mathcal{A}_d$} (\textbf{Meet in the Middle}): The finisher goes first to the point $(d, 0)$, where $d=(x^2+y^2)/(2x)$, and then towards the origin until it finds the failed starter (trajectory $P \rightarrow (d, 0) \rightarrow S \rightarrow T$). The point $(d,0)$ is the unique point on the line segment where the drones, both moving continuously, would meet simultaneously if the starter does not fail before time $d$.
\end{enumerate}

\begin{lemma}\mylabel{lm:Ad_vs_AT}
    $\text{CR}_{\mathcal{A}_d(x,y)} \leq \text{CR}_{\mathcal{A}_1(x,y)}$ if and only if the finisher starts within distance $1$ of the destination (i.e., $(x, y) \in \overline{D}(1, 1)$).
\end{lemma}
\begin{proof}
    First, we show that if the finisher's starting position $(x,y)$ is within a distance $1$ of the destination (i.e., $(x,y) \in \overline{D}(1,1)$ or equivalently $\sqrt{(x-1)^2+y^2} \leq 1$), then $\text{CR}_{\mathcal{A}_d}(t) \leq \text{CR}_{\mathcal{A}_1}(t)$ for any fail time $t$.
    Observe that if the finisher is executing algorithm $\mathcal{A}_d$ and $t \geq d$, then the starter fails only after the finisher encounters it at $(d,0)$, and so the algorithm is optimal.
    Thus, the competitive ratio for $\mathcal{A}_d$ can be written:
    \begin{align*}
        \text{CR}_{\mathcal{A}_d}(t) = \frac{A_d(t)}{\text{Opt}(t)} =
        \begin{cases}
            1 & \text{if } t \geq d \\
            \frac{\sqrt{(x-d)^2+y^2}+(d-t)+(1-t)}{\sqrt{(t-x)^2+y^2}-t+1} & \text{otherwise.}
        \end{cases} 
    \end{align*}
    Observe the max term is removed from $\text{Opt}(t)$ since
    \begin{align*}
        t &< d = \frac{x^2+y^2}{2x} \\
        2xt &< x^2 + y^2 \\ 
        t^2 &< x^2 - 2xt + t^2 + y^2 \\ 
        1 &< \sqrt{x^2 - 2xt + t^2 + y^2} - t + 1 .
    \end{align*}
    If $t \geq d$, then the statement is trivially true since $\mathcal{A}_d$ is optimal.
    Otherwise, since $(x, y) \in \overline{D}(1, 1)$, then $d = (x^2+y^2)/(2x) \leq 1$ and thus:
    \begin{align*}
        A_d(t) &= \sqrt{(x-d)^2+y^2} + d + 1 - 2t \leq \sqrt{(x-1)^2+y^2} + 2 - 2t = A_1(t).
    \end{align*}
    Now, we show that if the finisher starts a distance greater than $1$ from the destination, then $A_d(t) > A_1(t)$ for any fail time $t$ (and thus $\text{CR}_{\mathcal{A}_d} > \text{CR}_{\mathcal{A}_1}$).
    This follows from the fact that $d = (x^2+y^2)/(2x) > 1$ and thus:
    \begin{align*}
        A_d(t) &= \sqrt{(x-d)^2+y^2} + d + 1 - 2t > \sqrt{(x-1)^2+y^2} + 2 - 2t = A_1(t).
    \end{align*}%
\end{proof}

Lemma~\myref{lm:Ad_vs_AT} essentially tells us that we need only consider (among the candidate algorithms) $\mathcal{A}_0$ and $\mathcal{A}_d$ when the finisher starts \textit{inside} the disk and algorithms $\mathcal{A}_0$ and $\mathcal{A}_1$ when the finisher starts \textit{outside} the disk.
Note that, when $(x,y)$ is on the edge of the disk $\overline{D}(1,1)$ (i.e., $(x-1)^2+y^2=1$), then $d=1$ and so $\mathcal{A}_d$ and $\mathcal{A}_1$ are the same algorithm.

\subsection{Go to Last Point of Contact}
In this section, we derive the competitive ratio of Algorithm $\mathcal{A}_0$. 
Observe that, no matter the failure time of the starter, the algorithm takes time $1+\sqrt{x^2+y^2}$ to deliver the message.
This fact makes deriving the algorithm's competitive ratio rather simple.
\begin{theorem}
    \mylabel{thm:CR:A0}
    The competitive ratio of $\mathcal{A}_0$ is
    \begin{align*}
        \text{CR}_{\mathcal{A}_0} = \frac{1+\sqrt{x^2+y^2}}{\max\left\{1, \sqrt{(x-1)^2+y^2}\right\}} .
    \end{align*}
\end{theorem}
% \begin{appendixonly}
\begin{proof}
    Recall that the competitive ratio is given by the formula
    \begin{align*}
        \sup_{0 \leq t \leq 1} 
        \frac{1+\sqrt{x^2+y^2}}{\max\left\{1, \sqrt{(x-t)^2+y^2}+1-t\right\}}.
    \end{align*}
    Observe the denominator is non-increasing with respect to $t$ (on $t \in [0, 1]$) and attains a minimum at $t=1$.
    The competitive ratio, then is:
    \begin{align*}
        \frac{1+\sqrt{x^2+y^2}}{1+\max\left\{0, \sqrt{(x-1)^2+y^2}-1\right\}}
        &= \frac{1+\sqrt{x^2+y^2}}{\max\left\{1, \sqrt{(x-1)^2+y^2}\right\}}.
    \end{align*}
\end{proof}
% \end{appendixonly}

\subsection{Go to Destination}
We will now derive the competitive ratio of algorithm $\mathcal{A}_1$.
Unlike the delivery time of algorithm $\mathcal{A}_0$, the delivery time $A_1(t)=\sqrt{(x-1)^2+y^2}+2(1-t)$ (the time for the agent to go to $(1,0)$ and then backtrack until it finds the package at $(t,0)$ before returning to $(1,0)$ to complete the delivery) depends greatly on the fail time of the starter.
Thus, to find the competitive ratio of algorithm $\mathcal{A}_1$, we must find the worst-case fail time for a given finisher starting position $(x,y)$.
\begin{theorem}
    \mylabel{thm:CR:A1}
    Assume $(x,y) \not\in D(1,1)$ (i.e., $\sqrt{(x-1)^2+y^2} \geq 1$).
    The competitive ratio of $\mathcal{A}_1$ is \[\text{CR}_{\mathcal{A}_1}(\max\left\{t_1, 0\right\})\] where
    \begin{align*}
        t_1 = \begin{cases}
            1 - 3y/4 &\text{if }x=1 \\
            \frac{x^2 + y^2 + z_1(1-x) - 1 - z_1 \sqrt{x (x+z_1-2)+y^2-z_1+1}}{2 (x-1)} &\text{otherwise}
        \end{cases}
    \end{align*}
    and $z_1 = \sqrt{(x-1)^2+y^2}$.
\end{theorem}
% \begin{appendixonly}
\begin{proof}
    Note that $z_1$ is the distance from the finisher's initial location to the destination.
    Since \linebreak$\sqrt{(x-1)^2+y^2} \geq 1$, then observe $\sqrt{(x-t)^2+y^2} \geq t$ for all $0 \leq t \leq 1$.
    Thus, the delivery time for the optimal offline algorithm can be simplified (namely, the \textit{max} expression can be removed):
    \begin{align*}
        \text{Opt}(t) &= \max\left\{1, \sqrt{(x-t)^2+y^2} + 1 - t \right\}
        = \sqrt{(x-t)^2+y^2} + 1 - t .
    \end{align*}
    Thus the competitive ratio of $\mathcal{A}_1$ can be written:
    \begin{align*}
        \sup_{0 \leq t \leq 1} \text{CR}_{\mathcal{A}_1}(t)
        &= \sup_{0 \leq t \leq 1} \frac{A_1(t)}{\text{Opt}(t)} = \sup_{0 \leq t \leq 1} \frac{\sqrt{(x-1)^2+y^2}+2(1-t)}{\sqrt{(x-t)^2+y^2} + 1 - t} .
    \end{align*}
    To determine which value of $t$ maximizes $\text{CR}_{\mathcal{A}_1}(t)$, we can find the extremum of $\text{CR}_{\mathcal{A}_1}(t)$ by solving for $\frac{\partial}{\partial t} \text{CR}_{\mathcal{A}_1}(t) = 0$.
    The derivative of $\text{CR}_{\mathcal{A}_1}(t)$ with respect to $t$ is
    \begin{align*}
        &\frac{
            \left(1-\frac{t-x}{\sqrt{(t-x)^2+y^2}}\right) \left(\sqrt{(x-1)^2+y^2}+2(1-t)\right)
        }{
            \left(\sqrt{(t-x)^2+y^2}+1-t\right)^2
        } -~\frac{2}{\sqrt{(t-x)^2+y^2}+1-t} .
    \end{align*}
    Then by setting $\frac{\partial}{\partial t} \text{CR}_{\mathcal{A}_1}(t) = 0$ and solving for $t$, we get one solution when $x=1$:
    \begin{align*}
        t_{x=1} = 1- \frac{3y}{4}
    \end{align*}
    and two solutions\footnote{These solutions were derived with the aid of Mathematica (notebook available at 
    \url{https://anonymous.4open.science/r/faulty_delivery-D37F}
    % \url{https://github.com/jaredraycoleman/faulty_delivery/tree/main\#theorem-22}
    )} when $x \not= 1$:
    \begin{align*}
        t_- &= \frac{x^2 + y^2 + z_1(1-x) - 1 - z_1 \sqrt{x (x+z_1-2)+y^2-z_1+1}}{2 (x-1)} \\
        t_+ &= \frac{x^2 + y^2 + z_1(1-x) - 1 + z_1 \sqrt{x (x+z_1-2)+y^2-z_1+1}}{2 (x-1)}
    \end{align*}
    Observe, though, that $t_+ \geq 1$ when $x > 1$ and $t_+ < 0$ when $x < 1$, so $t_+$ cannot be a valid solution.
    Observe also that $\text{CR}_{\mathcal{A}_1}(1) = 1$ (the algorithm is optimal if the failure occurs at $(1,0)$) and $\text{CR}_{\mathcal{A}_1}(t_-) \geq 1$ (the competitive ratio, by construction, cannot be less than $1$).
    Then, since $\text{CR}_{\mathcal{A}_1}(t)$ is continuous, $\text{CR}_{\mathcal{A}_1}(t_-)$ must be a local extremum.
    Furthermore, since $t_- \leq 1$ for all $(x, y)$, $\text{CR}_{\mathcal{A}_1}(t)$ is decreasing on the interval $[\max\left\{t_-,0\right\},1]$.
    Thus, $\text{CR}_{\mathcal{A}_1}(t)$ is maximized at $t=\max\left\{0, t_-\right\}$ and the theorem is proved.
\end{proof}
% \end{appendixonly}

\begin{figure}[!ht]
    \centering
    \includegraphics[width=0.5\linewidth]{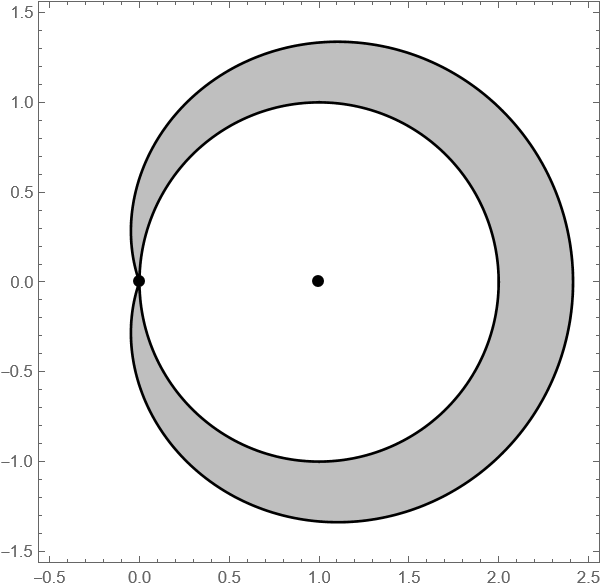}
    \caption{The worst-case fail time for the starter is $0$ \textit{except} in the gray shaded region, where the worst-case fail time is $t_-$ (from the proof of Theorem~\myref{thm:CR:A1}). Recall that we only consider $\mathcal{A}_1$ when the finisher starts outside of the Disk $D(1,1)$.}
    \mylabel{fig:A1_fail_at_0}
\end{figure}
There are essentially two cases that drive the competitive ratio of Algorithm $\mathcal{A}_1$: when $t_1$ (from the statement/proof of Theorem~\myref{thm:CR:A1}) is less than or equal to $0$ and when it is greater than $0$.
Observe the value of $t_1$ is determined by the finisher's starting position $(x,y)$.
When $t_1 \leq 0$ (i.e., when the finisher starts outside of the gray region shown in Figure~\myref{fig:A1_fail_at_0}), the competitive ratio is driven by the case where the starter fails at the origin (as soon as it loses contact with the command station).
When $t_1 > 0$, though, the worst-case scenario for $\mathcal{A}_1$ is actually when the starter fails at time $t_1$ (and thus, at location $(t_1,0)$).

\subsection{Meet in the Middle}\mylabel{sec:Ad}
In this section, we derive the competitive ratio of the last candidate algorithm, $\mathcal{A}_d$, where the finisher moves from its starting position to the point $d=(x^2+y^2)/(2x)$, then towards the origin until it finds the starter with the package before completing the delivery.
Recall from Lemma~\myref{lm:Ad_vs_AT} that we need only consider algorithm $\mathcal{A}_d$ inside the closed disk centered at $(1,0)$ with radius $1$, namely $\overline{D}(1,1)$.
Furthermore, observe that on the edge of this disk $d=1$ and so algorithms $\mathcal{A}_d$ and $\mathcal{A}_1$ are equal.
By construction, then, the finisher \textit{must} find the starter between the origin and the point $(d,0)$ since $(d,0)$ is the unique point on $\overline{ST}$ such that the distance between the origin and $(d,0)$ is equal to the distance from $(x,y)$ to $0$ (i.e., it is the point where the two drones would meet simultaneously if the starter does not fail). It's easy to derive that $d=(x^2+y^2)/(2x)$ by solving $\sqrt{\left(x-d\right)^2 + y^2} = d$ for $d$ using simple algebra.
% Indeed,
% \begin{align*}
%     \sqrt{\left(x-d\right)^2 + y^2}
%     &= \sqrt{\left(x-\frac{x^2+y^2}{2x}\right)^2 + y^2} \\
%     &= \sqrt{\left(\frac{x^2-y^2}{2x}\right)^2 + y^2} \\
%     &= \sqrt{\left(\frac{x^4-2x^2y^2+y^4}{4x^2}\right) + y^2} \\
%     &= \frac{\sqrt{\left(x^4-2x^2y^2+y^4\right) + 4x^2y^2}}{2x} \\
%     &= \frac{\sqrt{x^4+2x^2y^2+y^4}}{2x} \\
%     &= \frac{\sqrt{(x^2+y^2)^2}}{2x} = \frac{x^2+y^2}{2x} = d .
% \end{align*}
Thus, at time $d$, the finisher reaches point $(d,0)$ and the starter cannot have reached any point further than $(d,0)$.

\begin{theorem}\mylabel{thm:CR:Ad}
    Assume $(x,y) \in \overline{D}(1,1)$ (i.e., $\sqrt{(x-1)^2+y^2} \leq 1)$.
    The competitive ratio of $\mathcal{A}_d$ is
    \begin{align*}
        \text{CR}_{\mathcal{A}_d} = \begin{cases}
            \frac{x^2+y^2+x}{x\left(1+\sqrt{x^2+y^2}\right)} & \text{if } (x,y) \in \overline{D}(1/2,1/2) \\
            1+\frac{y^2}{x\left(\sqrt{x}+1\right)^2} & \text{otherwise.}
        \end{cases} 
    \end{align*}
\end{theorem}
\begin{appendixonly}
\begin{proof}
    We will derive the competitive ratio by finding the fail time $t$ that maximizes $\text{CR}_{\mathcal{A}_d}(t)$.
    Recall that $(d,0)$ is the location the finisher moves to in the $A_d$ algorithm (where $d = \frac{x^2+y^2}{2x}$).
    Observe that when $t \geq d$, the starter fails only \textit{after} it has met up with the finisher (and so the package delivery time is optimal and the competitive ratio is $1$).
    For this reason we need only consider fail times $t < d$.
    Observe also that,
    \begin{align*}
        t &< d = \frac{x^2+y^2}{2x} \\
        2xt &< x^2+y^2 \\
        0 &< x^2 - 2xt + y^2 \\
        t^2 &< x^2 - 2xt + t^2 + y^2 \\
        t^2 &< (x-t)^2 + y^2 \\
        0 &< \sqrt{(x-t)^2 + y^2} - t
    \end{align*}
    which means we can simplify 
    $$\text{Opt}(t) = \max\left\{1, \sqrt{(x-t)^2+y^2} + 1 - t \right\}$$
    and remove the max term, so 
    $$\text{Opt}(t) = \sqrt{(x-t)^2+y^2} + 1 - t.$$
    Then the competitive ratio of Algorithm $\mathcal{A}_d$ can be written:
    \begin{align*}
        \text{CR}_{\mathcal{A}_d} &= \sup_{0 \leq t < d} \text{CR}_{\mathcal{A}_d}(t) \\
        &= \sup_{0 \leq t < d} \vphantom{\rule{0pt}{5ex}} \frac{A_d(t)}{\text{Opt}(t)} \\
        &= \sup_{0 \leq t < d} \frac{\sqrt{(x-d)^2+y^2}+(d-t)+(1-t)}{\sqrt{(t-x)^2+y^2}-t+1} .
    \end{align*}
    We can find the extremum of $\text{CR}_{\mathcal{A}_d}(t)$ by solving for 
    $$\frac{\partial}{\partial t} \text{CR}_{\mathcal{A}_d}(t) = 0.$$
    The derivative of $\text{CR}_{\mathcal{A}_d}(t)$ with respect to $t$ is\footnote{We compute this derivative with the aid of Mathematica (notebook available at 
    \url{https://anonymous.4open.science/r/faulty_delivery-D37F}
    % \url{https://github.com/jaredraycoleman/faulty_delivery/tree/main\#theorem-23}
    ).}
    \begin{equation}\mylabel{eq:A_d:deriv}
        \frac{\partial}{\partial t} \text{CR}_{\mathcal{A}_d}(t) =
        \frac{
            t~b_0 + b_1(\sqrt{(x-t)^2+y^2}-x)
        }{
            x~\sqrt{(x-t)^2+y^2}(\sqrt{(x-t)^2+y^2}-t+1)
        }
    \end{equation}
    where $b_0=x(x-1)-y^2$, $b_1=x(x-1)-y^2$. Then, to find where the derivative is equal to $0$, let $b=-\frac{b_0}{b_1}=\frac{y^2-x(x-1)}{y^2+x(x-1)}$.
    \begin{align}
        \frac{\partial}{\partial t} \text{CR}_{\mathcal{A}_d}(t) &= 0 
        \nonumber \\
        \frac{t~b_0 + b_1(\sqrt{(x-t)^2+y^2}-x)}{x~\sqrt{(x-t)^2+y^2}(\sqrt{(x-t)^2+y^2}-t+1)} &= 0 \nonumber \\ 
        t~b_0 + b_1(\sqrt{(x-t)^2+y^2}-x) &= 0 \nonumber \\ 
        % t~b_0 &= - b_1\left(\sqrt{(x-t)^2+y^2}-x\right) \nonumber \\
        t\frac{-b_0}{b_1} &= \sqrt{(x-t)^2+y^2}-x \nonumber \\
        (bt + x)^2 &= (x-t)^2+y^2 \mylabel{eq:A_d:deriv=0}
    \end{align}

    Solving Equation~\myeqref{eq:A_d:deriv=0} for $t$, we get the two solutions:
    \begin{align*}
        t_- = \frac{-x(b+1) - \sqrt{(b+1)((b+1)x^2+(b-1)y^2)}}{b^2-1} \\
        t_+ = \frac{-x(b+1) + \sqrt{(b+1)((b+1)x^2+(b-1)y^2)}}{b^2-1}
    \end{align*}

    Upon replacing $b$ with $\frac{y^2-x(x-1)}{y^2+x(x-1)}$ and simplifying\footnote{We simplified these with the aid of Mathematica (notebook available at 
    \url{https://anonymous.4open.science/r/faulty_delivery-D37F}
    % \url{https://github.com/jaredraycoleman/faulty_delivery/tree/main\#theorem-23}
    ).}, we get the solutions
    % \marginnote{*this case statement is due to there being a $\sqrt{(x(x-1)+y^2)^2}$ in the expression}
    \begin{align*}
        t_- = \begin{cases} 
            \frac{x(x-1) + y^2}{2(x-\sqrt{x})} & \text{if } x^2+y^2 > x \\
            \frac{x(x-1) + y^2}{2(x+\sqrt{x})} & \text{otherwise}
        \end{cases} \\
        t_+ = \begin{cases} 
            \frac{x(x-1) + y^2}{2(x-\sqrt{x})} & \text{if } x^2+y^2 \leq x \\
            \frac{x(x-1) + y^2}{2(x+\sqrt{x})} & \text{otherwise}
        \end{cases}
    \end{align*}

    Observe, then, in both regions $x^2+y^2 > x$ and $x^2+y^2 \leq x$ (inside and outside of $D(1/2,1/2)$), $\text{CR}_{\mathcal{A}_d}(t)$ has a local maximum \textit{either} at $t^\prime=\frac{x(x-1) + y^2}{2(x + \sqrt{x})}$ or at  $t^{\prime\prime} = \frac{x(x-1) + y^2}{2(x - \sqrt{x})}$.
    We will use the following two claims to show that an upper bound on the competitive ratio for $\mathcal{A}_d$ is given by $t=t^\prime$ whenever $t^\prime > 0$ and $t=0$ otherwise.
    
    \vspace{\baselineskip}
    
    \begin{smallclaim}{$CR_{\mathcal{A}_d}(t^\prime)$ is a local maximum.}\mylabel{claim:Ad:t0localmax}
        Follows by second derivative test\footnote{See the Mathematica notebook available at 
        \url{https://anonymous.4open.science/r/faulty_delivery-D37F}
        %\url{https://github.com/jaredraycoleman/faulty_delivery/tree/main\#theorem-23-claim-1}
        for a detailed walk-through}, namely 
        $$CR_{\mathcal{A}_d}^{''}(t^\prime) < 0$$
        for all $x,y \geq 0$.
    \end{smallclaim}
    
    \begin{smallclaim}{$t^\prime < 1$.}\mylabel{claim:Ad:t0lessthan1}
        Since $(x,y) \in D(1,1)$, then $y^2 < 1 - (x-1)^2$ and since $t^\prime$ is increasing with respect to $y$
        \begin{align*}
            t^\prime &= \frac{x(x-1) + y^2}{2(x + \sqrt{x})} < \frac{x(x-1) + (1-(1-x)^2)}{2(x + \sqrt{x})} = \frac{1}{2 + 2/\sqrt{x}}.
        \end{align*}
        Furthermore, since $\frac{1}{2 + 2/\sqrt{x}}$ is increasing with respect to $x$ and $x<2$ (again, by the assumption that $(x,y) \in D(1,1)$), then 
        \begin{align*}
            \frac{1}{2 + 2/\sqrt{x}} < \frac{1}{2 + 2/\sqrt{2}} < 1
        \end{align*}
        and thus, $t^\prime < 1$.
    \end{smallclaim}

    Since $\text{CR}_{\mathcal{A}_d}(t)$ is decreasing on $(t^\prime, 1]$ (by Claim~\myref{claim:Ad:t0localmax}) and $CR_{\mathcal{A}_d}(t)$ is continuous on $t \leq 1$, then $\text{CR}_{\mathcal{A}_d}(t^\prime)$ is a local maximum on $[0,1]$ (and thus $\text{CR}_{\mathcal{A}_d}(t^{\prime\prime})$ is \textit{not} a local maximum on $[0,1]$).
    Furthermore, since $t^\prime < 1$ (by Claim~\myref{claim:Ad:t0lessthan1}), $\text{CR}_{\mathcal{A}_d}(t)$ is maximized at $t=t^\prime$ whenever $t^\prime > 0$ and $t=0$ otherwise.
    Plugging these values into $\text{CR}_{\mathcal{A}_d}(t)$ yields the value stated in the theorem.
\end{proof}
\end{appendixonly}

Similar to Algorithm $\mathcal{A}_1$, two cases drive Algorithm $\mathcal{A}_d$'s competitive ratio.
The first is when $t^\prime$ (from the proof of Theorem~\myref{thm:CR:Ad}) is less than $0$, whenever the starting position of the finisher is inside the disk $\overline{D}(1/2,1/2)$.
In this case, the worst-case failure time for the starter is $t=0$ (i.e., at the origin).
When the finisher starts outside of the disk $\overline{D}(1/2,1/2)$
(but inside the disk $\overline{D}(1,1)$, of course, since we only consider Algorithm $\mathcal{A}_d$ in this region), however, the worst-case failure time for the starter is $t=t^\prime$ (i.e., at location $(t^\prime,0)$, see Figure~\myref{fig:Ad_fail_at_0}).
\begin{figure}[h!]
    \centering
    \includegraphics[width=0.5\linewidth]{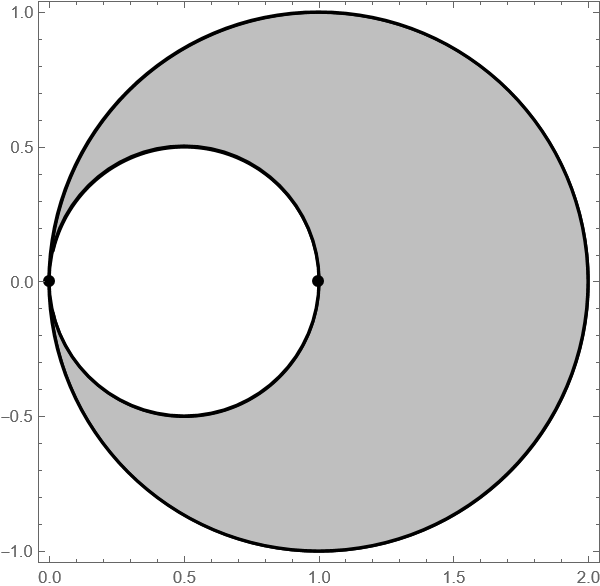}
    \caption{The worst-case fail time for the starter is $0$ \textit{except} in the gray shaded region, where the worst-case fail time is $t^\prime$ (from the proof of Theorem~\myref{thm:CR:Ad}). Recall that we only consider $\mathcal{A}_d$ when the finisher starts inside the Disk $\overline{D}(1,1)$.}\mylabel{fig:Ad_fail_at_0}
    \vspace{2em}
\end{figure}

\section{A Hybrid Algorithm}\mylabel{sec:hybrid}
For convenience, we summarize our main result by introducing Algorithm~\myref{alg:hybrid}, a hybrid algorithm which simply executes the best of $\mathcal{A}_0$, $\mathcal{A}_1$, and $\mathcal{A}_d$ given a finisher starting position $(x,y)$, and prove it to be optimal.

\begin{algorithm}[h!]
    \caption{A hybrid algorithm combining $\mathcal{A}_0$, $\mathcal{A}_1$, and $\mathcal{A}_d$}\mylabel{alg:hybrid}
    \begin{algorithmic}[1]
        \State \textbf{input}: Finisher starting position $(x,y)$
        \If{$(x-1)^2+y^2>1$} \Comment{The starter starts outside of $D(1,1)$}
            \If{$\text{CR}_{\mathcal{A}_0(x,y)} \leq \text{CR}_{\mathcal{A}_1(x,y)}$}
                \State Execute Algorithm $\mathcal{A}_0$
            \Else
                \State Execute Algorithm $\mathcal{A}_1$
            \EndIf
        \Else
            \If{$\text{CR}_{\mathcal{A}_0(x,y)} \leq \text{CR}_{\mathcal{A}_d(x,y)}$}
                \State Execute Algorithm $\mathcal{A}_0$
            \Else
                \State Execute Algorithm $\mathcal{A}_d$
            \EndIf
        \EndIf
    \end{algorithmic}
\end{algorithm}

In fact, we will show that Algorithm~\myref{alg:hybrid} is optimal by proving that any other algorithm $A_a$ which moves first to a position $(a,0)$ such that $0 < a < 1$ and $a \neq d$ is always worse than at least one of the candidate algorithms $\mathcal{A}_0$, $\mathcal{A}_1$, or $\mathcal{A}_d$.
Recall from Lemma~\myref{lm:alg_class} that we do not need to consider any other algorithms.
Leveraging Lemma~\myref{lm:Ad_vs_AT}, we know that we only need consider $\mathcal{A}_0$ and $\mathcal{A}_d$ when the finisher's starting position $(x,y)$ is inside the disk $\overline{D}(1,1)$ and algorithms $\mathcal{A}_0$ and $\mathcal{A}_1$ when it is outside the disk (recall that $\mathcal{A}_d$ and $\mathcal{A}_1$ are equivalent on the border of $\overline{D}(1,1)$).
The following lemmas cover each of these cases.
\begin{lemma}\mylabel{lm:lower_inside}
    For any $(x,y) \in \overline{D}(1, 1)$, 
    % either $\text{CR}_{\mathcal{A}_d} \leq \text{CR}_{\mathcal{A}_a}$ or $\text{CR}_{\mathcal{A}_0} \leq \text{CR}_{\mathcal{A}_a}$.
    $\min\left\{\text{CR}_{\mathcal{A}_d}, \text{CR}_{\mathcal{A}_0} \right\} \leq \text{CR}_{\mathcal{A}_a}$ for all $a \in \overline{ST}$.
\end{lemma}
\begin{appendixonly}
\begin{proof}
We will prove the lemma by examining various cases.
Let $t_d$ denote the worst-case starter fail time for algorithm $\mathcal{A}_d$ for the finisher starting position $(x,y)$.
Recall that $(a,0)$, as defined above, is the first point on $\overline{ST}$ that the finisher reaches by executing algorithm $\mathcal{A}_a$.

\noindent\textbf{Case 1}: $a \geq d$. Consider the case where the starter fails at $(t_d, 0)$. Then the competitive ratio of the algorithm is:
\begin{align*}
    \text{CR}_{\mathcal{A}_a} &\geq \frac{\sqrt{(x-a)^2 + y^2} + (a-t_d)+(1-t_d)}{\max\left\{1,\sqrt{(x-t_d)^2+y^2} + 1 - t_d\right\}}  \\
         &\geq \frac{\sqrt{(x-d)^2 + y^2} + (d-t_d)+(1-t_d)}{\max\left\{1,\sqrt{(x-t_d)^2+y^2} + 1 - t_d\right\}}\\
         &= \text{CR}_{\mathcal{A}_d} .
\end{align*}

\noindent\textbf{Case 2}: $a < d$ and finisher visits $S$ before visiting $T$. Set failure to be $T$ then $\text{CR}_{\mathcal{A}_0} \leq \text{CR}_{\mathcal{A}_a}$. 
Indeed,
\begin{align*}
    \text{CR}_{\mathcal{A}_a} \geq \frac{\sqrt{(x-a)^2 + y^2} + a + 1}{\sqrt{(x-1)^2+y^2}} \geq \frac{\sqrt{x^2 + y^2} + 1}{\sqrt{(x-1)^2+y^2}} = \text{CR}_{\mathcal{A}_0} .
\end{align*}
\textbf{Case 3}: $a < d$ and finisher visits $1$ before visiting $0$.

\noindent\textbf{Case 3a}: finisher visits $t_d$ after visiting $1$.
Consider the case where the starter fails at time $t_d$.
Then the competitive ratio of $\mathcal{A}_a$ can be written:
\begin{align*}
    \text{CR}_{\mathcal{A}_a} &\geq \frac{\sqrt{(x-a)^2 + y^2} + (1-a) + 2(1-t_d)}{\max\left\{1,\sqrt{(x-t_d)^2+y^2} + 1 - t_d\right\}} \\
         &\geq \frac{\sqrt{(x-d)^2 + y^2} + (d-t_d)+(1-t_d)}{\max\left\{1,\sqrt{(x-t_d)^2+y^2} + 1 - t_d\right\}}\\
         &= \text{CR}_{\mathcal{A}_d} .
\end{align*}
\textbf{Case 3b}: finisher visits $t_d$ before visiting $1$.
If $t_d=0$, then $\text{CR}_{\mathcal{A}_0} \leq \text{CR}_{\mathcal{A}_a}$ by Case $1$.
Otherwise, consider the case where $t=0$ (i.e., the starter fails at the origin).
The competitive ratio of $\mathcal{A}_a$, then, is:
\begin{align*}
    \text{CR}_{\mathcal{A}_a }
    &\geq \frac{\sqrt{(x-a)^2+y^2}+|a-t_d|+3-t_d}{\sqrt{x^2+y^2}+1} \\
    \text{CR}_{\mathcal{A}_a } &\geq \frac{\sqrt{(x-t_d)^2+y^2}+3-t_d}{\sqrt{x^2+y^2}+1}
\end{align*}
Then, by substituting $t_d = \frac{x(x-1)+y^2}{2(\sqrt{x} + x)}$ (this is the worst-case fail time for algorithm $\mathcal{A}_d$ whenever $t_d > 0$, see the proof for Theorem~\myref{thm:CR:Ad}), we get
\begin{align*}
    &\frac{\sqrt{(x-t_d)^2+y^2}+3-t_d}{\sqrt{x^2+y^2}+1} \\
    &= \frac{
        \sqrt{\left(x-\frac{(x-1) x+y^2}{2\left(x+\sqrt{x}\right)}\right)^2+y^2}
        + 3 - \frac{(x-1) x+y^2}{2 \left(x+\sqrt{x}\right)}
    }{\sqrt{x^2+y^2}+1} .
\end{align*}
By expanding the term under the radical in the numerator and collecting like terms, we arrive at 
\begin{align*}
    &\frac{
        \frac{
            \sqrt{y^2 4 (x+\sqrt{x})^2 + [2x(x+\sqrt{x}) - (x(x-1)+y^2)]^2}
        }{2(x+\sqrt{x})}
        + 3 - \frac{(x-1) x+y^2}{2 \left(x+\sqrt{x}\right)}
    }{\sqrt{x^2+y^2}+1}
    % &~~~~~~~~=\frac{
    %     \frac{
    %         \sqrt{y^2 4 (x+\sqrt{x})^2 + [x^2+2x\sqrt{x}+x-y^2]^2}
    %     }{2(x+\sqrt{x})}
    %     + 3 - \frac{(x-1) x+y^2}{2 \left(x+\sqrt{x}\right)}
    % }{\sqrt{x^2+y^2}+1}
\end{align*}
which can then be readily simplified (by continuing to multiply out and collect like terms) to: 
\begin{align*}
    \frac{
        \frac{
            \sqrt{y^2 4 (x+\sqrt{x})^2 + [(x+\sqrt{x})^2-y^2]^2}
        }{2(x+\sqrt{x})}
        + 3 - \frac{(x-1) x+y^2}{2 \left(x+\sqrt{x}\right)}
    }{\sqrt{x^2+y^2}+1} .
\end{align*}
Then, by factoring the expression under the radical in the numerator, we can simplify further
\begin{align}
    &\frac{
        \frac{
            \sqrt{4(x+\sqrt{x})^2 y^2 + ((x+\sqrt{x})^2)^2 - 2 (x+\sqrt{x})^2 y^2 + y^4}
        }{2(x+\sqrt{x})}
        + 3 - \frac{(x-1) x+y^2}{2 \left(x+\sqrt{x}\right)}
    }{\sqrt{x^2+y^2}+1} \nonumber \\
    &~~~~~~~~=\frac{
        \frac{
            \sqrt{(y^2+(x+\sqrt{x})^2)^2}
        }{2(x+\sqrt{x})}
        + 3 - \frac{(x-1) x+y^2}{2 \left(x+\sqrt{x}\right)}
    }{\sqrt{x^2+y^2}+1} \nonumber \\
    &~~~~~~~~=\frac{
        \frac{y^2+(x+\sqrt{x})^2}{2(x+\sqrt{x})}
        + 3 - \frac{(x-1) x+y^2}{2 \left(x+\sqrt{x}\right)}
    }{\sqrt{x^2+y^2}+1}  \nonumber
\end{align}
and finally arrive at the final, rather simple, expression:
\begin{align}
    &\frac{
        3 + \frac{y^2 + x^2 + 2x\sqrt{x} + x - (x^2 - x + y^2)}{2(x + \sqrt{x})}
    }{\sqrt{x^2+y^2}+1} =\frac{3+\sqrt{x}}{\sqrt{x^2+y^2}+1} . \mylabel{eq:inner_cra_upper:1}
\end{align}
Since $t_d > 0$, then $\text{CR}_{\mathcal{A}_d} = 1 + \frac{y^2}{x(\sqrt{x}+1)^2}$ by Theorem~\myref{thm:CR:Ad} and we have
\begin{align}
    \text{CR}_{\mathcal{A}_a} 
    &\geq \frac{3+\sqrt{x}}{\sqrt{x^2+y^2}+1} \mylabel{eq:inner_cra_upper:2}\\ 
    %\frac{3+\sqrt{x}}{\sqrt{x^2+y^2}+1} 
    &\geq \frac{3+\sqrt{x}}{\sqrt{2x}+1} \mylabel{eq:inner_cra_upper:3}\\
    %\frac{3+\sqrt{x}}{\sqrt{2x}+1} 
    &\geq \frac{3+2\sqrt{x}}{(1+\sqrt{x})^2} \mylabel{eq:inner_cra_upper:4}\\
    %\frac{3+2\sqrt{x}}{(1+\sqrt{x})^2} 
    &\geq 1 + \frac{y^2}{x(\sqrt{x}+1)^2}  \mylabel{eq:inner_cra_upper:5}
    = \text{CR}_{\mathcal{A}_d} 
\end{align}
where~\myeqref{eq:inner_cra_upper:2} follows from~\myeqref{eq:inner_cra_upper:1}, \myeqref{eq:inner_cra_upper:3} and~\myeqref{eq:inner_cra_upper:5} follow since $x \in D(1,1)$ (and thus $y^2 \leq 1 - (x-1)^2$ and $x > 0$), and~\myeqref{eq:inner_cra_upper:4} is easy to see by expanding and simplifying: 
\begin{align*}
    \frac{3+\sqrt{x}}{\sqrt{2x}+1} &\geq \frac{3+2\sqrt{x}}{(1+\sqrt{x})^2} \\
    x^{3/2}+5 x+7 \sqrt{x}+3 &\geq 2 \sqrt{2} x+3 \sqrt{2} \sqrt{x}+2 \sqrt{x}+3 \\
    \sqrt{x} \left(x+\left(5-2 \sqrt{2}\right) \sqrt{x}+5-3\sqrt{2}\right) &\geq 0
\end{align*}
which is trivially true since $x>0$, $5 > 2 \sqrt{2}$, and $5\geq 3 \sqrt{2}$.
\end{proof}
\end{appendixonly}

Again, by Lemma~\myref{lm:Ad_vs_AT}, when the finisher starts outside of the disk $D(1,1)$, we need only consider $\mathcal{A}_0$ and $\mathcal{A}_1$ (recall that, on the edge of the disk $\overline{D}(1,1)$, $\mathcal{A}_1$ and $\mathcal{A}_d$ are equivalent).
\begin{lemma}\mylabel{lm:lower_outside}
    For any $(x,y) \not\in D(1, 1)$, either $\text{CR}_{\mathcal{A}_1} \leq \text{CR}_{\mathcal{A}_a}$ or $\text{CR}_{\mathcal{A}_0} \leq \text{CR}_{\mathcal{A}_a}$ for all $a \in \overline{ST}$.
\end{lemma}
% \begin{appendixonly}
\begin{proof}
    We will prove the lemma by examining various cases.
    Let $t_1$ denote the worst-case fail time of $A_1$ for the finisher starting position $(x,y)$.
    
    \noindent\textbf{Case 1}: The finisher visits $S$ before visiting $T$.
    Then let $t=1$, and the competitive ratio of $\mathcal{A}_a$ is
    \begin{align*}
        \text{CR}_{\mathcal{A}_a} &\geq \frac{\sqrt{(x-a)^2 + y^2} + a + 1}{\max\left\{1,\sqrt{(x-1)^2+y^2}\right\}} 
        \geq \frac{\sqrt{x^2 + y^2} + 1}{\max\left\{1,\sqrt{(x-1)^2+y^2}\right\}}
        = \text{CR}_{\mathcal{A}_0} .
    \end{align*}

    \noindent\textbf{Case 2:} The finisher visits $T$ before visiting $S$. Then we break the analysis into sub-cases. Let $b$ be the smallest value such that the finisher visits $(b,0)$ before visiting $(1,0)$.
    
    \noindent\textbf{Case 2a:} $b > t_1$. Then let consider the case where the starter fails at time $t_1$. The the competitive ratio of $\mathcal{A}_a$ is
    \begin{align*}
        \text{CR}_{\mathcal{A}_a} &\geq \frac{\sqrt{(x-a)^2 + y^2} + (1-a) + 2(1-t_1)}{\max \left\{1,\sqrt{(x-t_1)^2+y^2} + (1-t_1)\right\} }  \\
        &\geq \frac{\sqrt{(x-1)^2 + y^2} + 2(1-t_1)}{\max\left\{1,\sqrt{(x-t_1)^2+y^2} + (1-t_1)\right\} } = \text{CR}_{\mathcal{A}_1} .
    \end{align*}
    
    \noindent\textbf{Case 2b:} $b \leq t_1$. Then consider the case where the starter fails at time $b-\epsilon$ (where is $\epsilon$ is arbitrarily small), and:
    \begin{align}
        \text{CR}_{\mathcal{A}_a} &\geq \lim_{\epsilon \rightarrow 0} \frac{\sqrt{(x-a)^2+y^2}+(a-b)+3(1-(b-\epsilon))}{\sqrt{(x-(b-\epsilon))^2+y^2}+1-b} \notag \\
        &\geq \lim_{\epsilon \rightarrow 0} \frac{\sqrt{(x-b)^2+y^2}+3(1-(b-\epsilon))}{\sqrt{(x-(b-\epsilon))^2+y^2}+1-b} \nonumber \\
        &\geq \frac{\sqrt{(x-b)^2+y^2}+3(1-b)}{\sqrt{(x-b)^2+y^2}+1-b} \mylabel{eq:outer_cra_upper:2} \\
        &\geq \frac{\sqrt{(x-t_1)^2+y^2}+3(1-t_1)}{\sqrt{(x-t_1)^2+y^2}+1-t_1} \mylabel{eq:outer_cra_upper:3} \\
        &\geq \frac{\sqrt{(x-1)^2+y^2}+2(1-t_1)}{\sqrt{(x-t_1)^2+y^2}+1-t_1} = \text{CR}_{\mathcal{A}_1} \notag
    \end{align}
    where $\myeqref{eq:outer_cra_upper:2}$ follows from the triangle inequality and $\myeqref{eq:outer_cra_upper:3}$ follows since $\frac{\sqrt{(x-b)^2+y^2}+3(1-b)}{\sqrt{(x-b)^2+y^2}+1-b}$ is decreasing with respect to $b$ (and thus is minimized at $b \rightarrow t_1$).
\end{proof}
% \end{appendixonly}

\begin{theorem}
    Algorithm~\myref{alg:hybrid} is optimal.
\end{theorem}
\begin{proof}
    Follows directly from Lemmas~\myref{lm:lower_inside} and~\myref{lm:lower_outside}.
\end{proof}

\section{Discussion}\mylabel{sec:discussion}
\begin{figure*}[!ht]
    \centering
    \begin{subfigure}{\linewidth}
    \centering
        \includegraphics[width=0.6\linewidth]{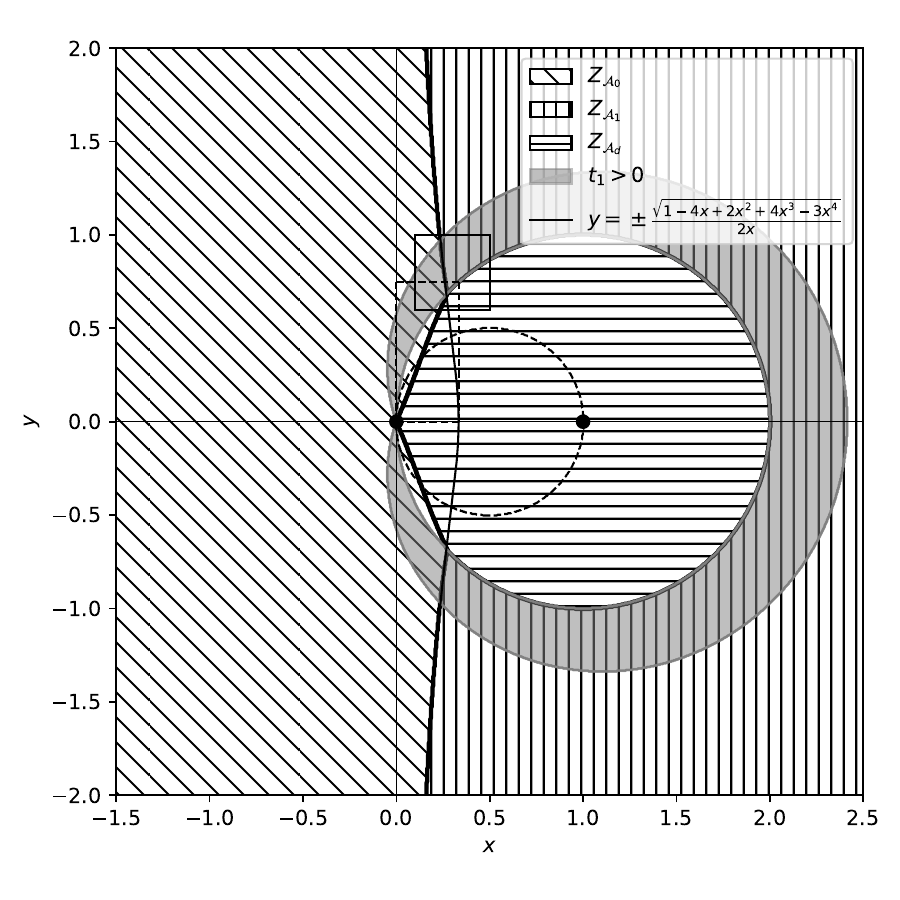}
        \caption{Plot showing which candidate algorithm has the least competitive ratio given the finisher's starting position $(x,y)$.}
        \mylabel{fig:regions:all}
    \end{subfigure}%

    \vspace{2em}%

    \begin{subfigure}{0.48\linewidth}
    \centering
      \includegraphics[trim={3.9cm 0 3.9cm 0},width=0.47\linewidth]{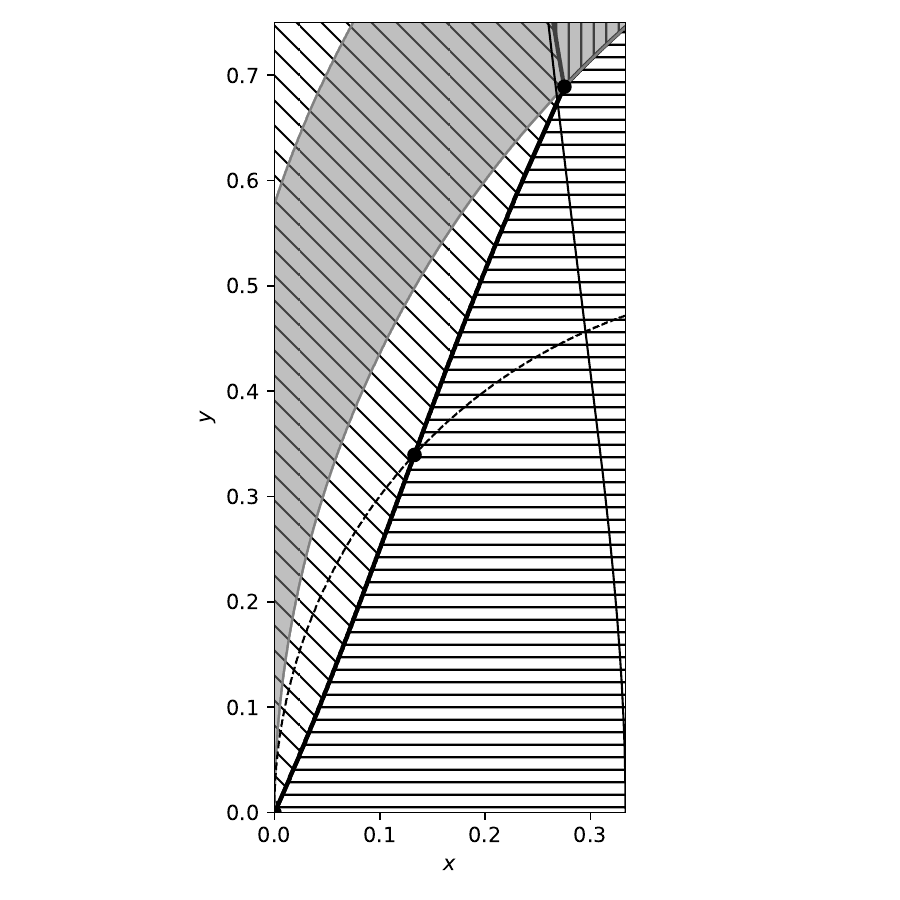}
      \caption{The region bounded by the dash-bordered rectangular box in Figure~\myref{fig:regions:all}.}
      \mylabel{fig:regions:box_1}
    \end{subfigure}%
    % \vspace{2em}%
    \hspace{0.02\linewidth}
    \begin{subfigure}{0.48\linewidth}
    \centering
      \includegraphics[width=\linewidth]{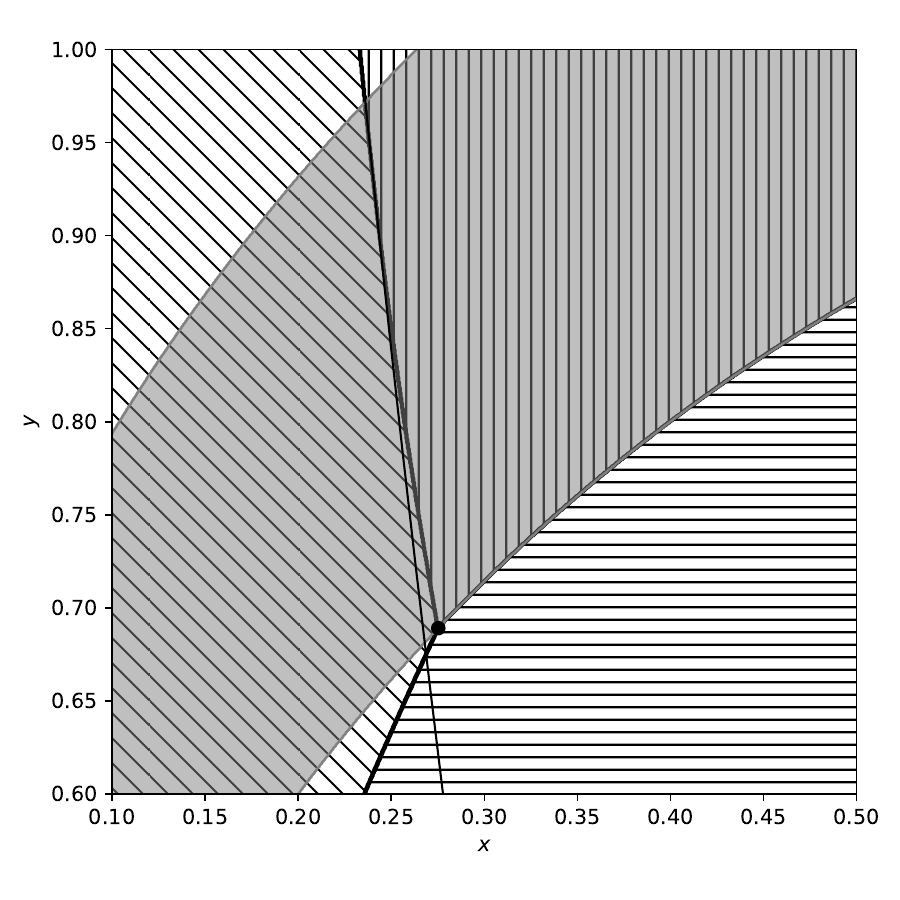}
      \caption{The region bounded by the solid-bordered square box in Figure~\myref{fig:regions:all}.
      {\color{white} xxxxx xxxxx xxxxx xxxxx xxxxx xxxxx} % dummy text to align captions better
      }\mylabel{fig:regions:box_0}
    \end{subfigure}%

    \vspace{1em}%
    
  \caption{Region plots with highlighted areas}
  \mylabel{fig:regions}
\end{figure*}
The competitive ratio of each of the candidate algorithms, and therefore the hybrid algorithm, depends on the starting position of the finisher.
Let $Z_{\mathcal{A}_0}$ denote the set of points $(x,y)$ such that Algorithm~\myref{alg:hybrid} executes $\mathcal{A}_0$ (i.e., $(x-1)^2+y^2>1$ and $\text{CR}_{\mathcal{A}_0(x,y)} \leq \text{CR}_{\mathcal{A}_1(x,y)}$, or else $(x-1)^2+y^2 \leq 1$ and $\text{CR}_{\mathcal{A}_0(x,y)} \leq \text{CR}_{\mathcal{A}_d(x,y)}$).
Similarly, let $Z_{\mathcal{A}_1}$ denote the set of points $(x,y)$ such that Algorithm~\myref{alg:hybrid} executes $\mathcal{A}_1$ (i.e., $(x-1)^2+y^2>1$ and $\text{CR}_{\mathcal{A}_0(x,y)} > \text{CR}_{\mathcal{A}_1(x,y)}$) 
and $Z_{\mathcal{A}_1}$ the set of points $(x,y)$ such that Algorithm~\myref{alg:hybrid} executes $\mathcal{A}_d$ (i.e., $(x-1)^2+y^2 \leq 1$ and $\text{CR}_{\mathcal{A}_0(x,y)} > \text{CR}_{\mathcal{A}_d(x,y)}$).
In other words, for a finisher starting position $(x,y)$, the hybrid algorithm executes algorithm $\mathcal{A}_0$ iff $(x,y) \in Z_{\mathcal{A}_0}$, $\mathcal{A}_1$ iff $(x,y) \in Z_{\mathcal{A}_1}$, and $\mathcal{A}_d$ iff $(x,y) \in Z_{\mathcal{A}_d}$.
Observe the regions $Z_{\mathcal{A}_0}$, $Z_{\mathcal{A}_1}$, and $Z_{\mathcal{A}_d}$ are disjoint and their union comprises the entire plane.
Figure~\myref{fig:regions} depicts these regions.

Observe if $(x,y)$ is inside the disk $D(1,1)$, either $\mathcal{A}_0(x,y)$ or $\mathcal{A}_d(x,y)$ has the least competitive ratio.
Otherwise, if $(x,y) \not\in \overline{D}(1,1)$ either $\mathcal{A}_0(x,y)$ or $\mathcal{A}_1(x,y)$ has the least competitive ratio.
This is consistent, of course, with Lemma~\myref{lm:Ad_vs_AT}.
Inside each of these regions, there are even more interesting things going on.
Examine the bang-bang curve (i.e., the separating curve) between the regions $Z_{\mathcal{A}_0}$ and $Z_{\mathcal{A}_d}$ inside the disk $D(1/2,1/2)$.
Recall that the worst-case fail time for $\mathcal{A}_d(x,y)$ is $t_d=\frac{x(x-1)+y^2}{2(x+\sqrt{x})}$ if $(x,y)\in D(1/2,1/2)$ and $0$ otherwise.
This is what causes the change of inflection when the bang-bang curve leaves the disk $D(1/2,1/2)$.

Now, look at the bang-bang curve between $Z_{\mathcal{A}_0}$ and $Z_{\mathcal{A}_1}$  outside of the disk $D(1,1)$.
Recall the worst-case fail time for $\mathcal{A}_1(x,y)$ is
\begin{align*}
    t_1 = \begin{cases}
        1 - 3y/4 &\text{if }x=1 \\
        \frac{x^2 + y^2 + z_1(1-x) - 1 - z_1 \sqrt{x (x+z_1-2)+y^2-z_1+1}}{2 (x-1)} &\text{otherwise}
    \end{cases}
\end{align*}
(where $z_1 = \sqrt{(x-1)^2+y^2}$) if $t_1 > 0$ and $0$ otherwise.
Outside of this region, we can actually find a closed-form equation for the curve separating $\mathcal{A}_0$ and $\mathcal{A}_1$ since the competitive ratios for $\text{CR}_{\mathcal{A}_0}$ and $\text{CR}_{\mathcal{A}_1}$ are relatively simple.
Indeed, we can easily find where the competitive ratios are equal for this case:
\begin{align*}
    \text{CR}_{\mathcal{A}_1(x,y)}(0) &= \text{CR}_{\mathcal{A}_0(x,y)} \\
    \frac{\sqrt{(x-1)^2+y^2}+2}{\sqrt{x^2+y^2} + 1} &= \frac{\sqrt{x^2+y^2}+1}{\sqrt{(x-1)^2+y^2}} .
\end{align*}
By cross-multiplying the right-hand side and simplifying, we obtain:
\begin{align*}
    % 2\sqrt{(x-1)^2+y^2} + (x-1)^2+y^2 &= 1 + 2\sqrt{x^2+y^2} + x^2 + y^2 \\
    % \sqrt{(x-1)^2+y^2} &= \sqrt{x^2+y^2} + x \\
    % (x-1)^2+y^2 &= x^2+y^2 + 2x\sqrt{x^2+y^2} + x^2 \\
    -2x+1 &= 2x\sqrt{x^2+y^2} + x^2
    % \frac{1-2x-x^2}{2x} &= \sqrt{x^2+y^2} \\
\end{align*}
From this, it is easy to arrive at:
\begin{align*}
    y &= \pm \sqrt{\left(\frac{1-2x-x^2}{2x}\right)^2 - x^2} = \pm \frac{\sqrt{1-4x+2x^2+4x^3-3x^4}}{2x}
\end{align*}
Unfortunately, this is the only bang-bang curve for which we were able to find a closed form equation for.
Observe in Figure~\myref{fig:regions:box_0} how the curve diverges from the curve given by $y=\frac{\sqrt{1-4x+2x^2+4x^3-3x^4}}{2x}$ when $t_1 > 0$ (inside the shaded region).

\section{Conclusion}\mylabel{sec:conclusion}
In this paper, we present a competitively optimal algorithm for two-agent delivery in the plane with a single faulty agent.
We show that the competitive ratio of the algorithm depends relative positioning of the two agents and the destination.
An upper bound of $3$ on the competitive ratio over all points $(x,y)$ is straightforward to prove (both $\text{CR}_{\mathcal{A}_0(x,y)}$ and $\text{CR}_{\mathcal{A}_1(x,y)}$ have a maximum of $3$).
Numerical calculations indicate the maximum competitive ratio is given by the case where $\text{CR}_{\mathcal{A}_0(x,y)} = \text{CR}_{\mathcal{A}_1(x,y)} = \text{CR}_{\mathcal{A}_d(x,y)}$, which is approximately $1.74197$ at $x \approx 0.275257, y \approx 0.689019$.
The results presented in this paper introduce a number of interesting questions and avenues for future work.
First, we assume the starter can only move directly toward the destination.
It would be interesting to see how the competitive ratio changes if the starter can move in any direction (e.g., toward the finisher).
Second, we could remove the assumption that the starter/finisher move at the same speed.
In this case, the finisher might be able to deliver the package faster by participating even if the starter doesn't fail.
Finally, we could extend the problem to consider multiple (potentially faulty) agents (i.e., what happens if the finisher also fails?).

\newpage
\bibliographystyle{plainurl}
\bibliography{main}

\begin{thebibliography}{10}

\bibitem{faulty_gathering}
Noa Agmon and David Peleg.
\newblock Fault-tolerant gathering algorithms for autonomous mobile robots.
\newblock {\em {SIAM} J. Comput.}, 36(1):56--82, 2006.
\newblock \href {https://doi.org/10.1137/050645221} {\path{doi:10.1137/050645221}}.

\bibitem{search_gal_survey}
Steve Alpern and Shmuel Gal.
\newblock {\em The theory of search games and rendezvous}, volume~55 of {\em International series in operations research and management science}.
\newblock Kluwer, 2003.

\bibitem{del_carvalho_graph}
Iago~A. Carvalho, Thomas Erlebach, and Kleitos Papadopoulos.
\newblock On the fast delivery problem with one or two packages.
\newblock {\em J. Comput. Syst. Sci.}, 115:246--263, 2021.
\newblock \href {https://doi.org/10.1016/j.jcss.2020.09.002} {\path{doi:10.1016/j.jcss.2020.09.002}}.

\bibitem{sd_coleman_line}
Jared Coleman, Lorand Cheng, and Bhaskar Krishnamachari.
\newblock Search and rescue on the line.
\newblock In Sergio Rajsbaum, Alkida Balliu, Joshua~J. Daymude, and Dennis Olivetti, editors, {\em Structural Information and Communication Complexity}, pages 297--316, Cham, 2023. Springer Nature Switzerland.

\bibitem{del_coleman_plane}
Jared Coleman, Evangelos Kranakis, Danny Krizanc, and Oscar Morales{-}Ponce.
\newblock Message delivery in the plane by robots with different speeds.
\newblock In Colette Johnen, Elad~Michael Schiller, and Stefan Schmid, editors, {\em Stabilization, Safety, and Security of Distributed Systems - 23rd International Symposium, {SSS} 2021, Virtual Event, November 17-20, 2021, Proceedings}, volume 13046 of {\em Lecture Notes in Computer Science}, pages 305--319. Springer, 2021.
\newblock \href {https://doi.org/10.1007/978-3-030-91081-5\_20} {\path{doi:10.1007/978-3-030-91081-5\_20}}.

\bibitem{del_pony_express}
Jared Coleman, Evangelos Kranakis, Danny Krizanc, and Oscar Morales{-}Ponce.
\newblock The pony express communication problem.
\newblock In Paola Flocchini and Lucia Moura, editors, {\em Combinatorial Algorithms - 32nd International Workshop, {IWOCA} 2021, Ottawa, ON, Canada, July 5-7, 2021, Proceedings}, volume 12757 of {\em Lecture Notes in Computer Science}, pages 208--222. Springer, 2021.
\newblock \href {https://doi.org/10.1007/978-3-030-79987-8\_15} {\path{doi:10.1007/978-3-030-79987-8\_15}}.

\bibitem{faulty_patrolling}
Jurek Czyzowicz, Leszek Gasieniec, Adrian Kosowski, Evangelos Kranakis, Danny Krizanc, and Najmeh Taleb.
\newblock When patrolmen become corrupted: Monitoring a graph using faulty mobile robots.
\newblock {\em Algorithmica}, 79(3):925--940, 2017.
\newblock \href {https://doi.org/10.1007/s00453-016-0233-9} {\path{doi:10.1007/s00453-016-0233-9}}.

\bibitem{faulty_evacuation}
Jurek Czyzowicz, Konstantinos Georgiou, Maxime Godon, Evangelos Kranakis, Danny Krizanc, Wojciech Rytter, and Michal Wlodarczyk.
\newblock Evacuation from a disc in the presence of a faulty robot.
\newblock In Shantanu Das and S{\'{e}}bastien Tixeuil, editors, {\em Structural Information and Communication Complexity - 24th International Colloquium, {SIROCCO} 2017, Porquerolles, France, June 19-22, 2017, Revised Selected Papers}, volume 10641 of {\em Lecture Notes in Computer Science}, pages 158--173. Springer, 2017.
\newblock \href {https://doi.org/10.1007/978-3-319-72050-0\_10} {\path{doi:10.1007/978-3-319-72050-0\_10}}.

\bibitem{faulty_search}
Jurek Czyzowicz, Evangelos Kranakis, Danny Krizanc, Lata Narayanan, and Jaroslav Opatrny.
\newblock Search on a line with faulty robots.
\newblock {\em Distributed Comput.}, 32(6):493--504, 2019.
\newblock \href {https://doi.org/10.1007/s00446-017-0296-0} {\path{doi:10.1007/s00446-017-0296-0}}.

\bibitem{del_drone_graph}
Thomas Erlebach, Kelin Luo, and Frits C.~R. Spieksma.
\newblock Package delivery using drones with restricted movement areas.
\newblock In Sang~Won Bae and Heejin Park, editors, {\em 33rd International Symposium on Algorithms and Computation, {ISAAC} 2022, December 19-21, 2022, Seoul, Korea}, volume 248 of {\em LIPIcs}, pages 49:1--49:16. Schloss Dagstuhl - Leibniz-Zentrum f{\"{u}}r Informatik, 2022.
\newblock \href {https://doi.org/10.4230/LIPIcs.ISAAC.2022.49} {\path{doi:10.4230/LIPIcs.ISAAC.2022.49}}.

\bibitem{DBLP:conf/algosensors/GeorgiouKLPP19}
Konstantinos Georgiou, Evangelos Kranakis, Nikos Leonardos, Aris Pagourtzis, and Ioannis Papaioannou.
\newblock Optimal circle search despite the presence of faulty robots.
\newblock In {\em {Algosensors} 2019, Munich, Germany, September 12-13, 2019}, volume 11931 of {\em LNCS}, pages 192--205. Springer, 2019.
\newblock \href {https://doi.org/10.1007/978-3-030-34405-4\_11} {\path{doi:10.1007/978-3-030-34405-4\_11}}.

\bibitem{drone_survey}
Asif~Mahmud Raivi, S.~M.~Asiful Huda, Muhammad~Morshed Alam, and Sangman Moh.
\newblock Drone routing for drone-based delivery systems: {A} review of trajectory planning, charging, and security.
\newblock {\em Sensors}, 23(3):1463, 2023.
\newblock \href {https://doi.org/10.3390/s23031463} {\path{doi:10.3390/s23031463}}.

\bibitem{faulty_flocking}
Yan Yang, Samia Souissi, Xavier D{\'{e}}fago, and Makoto Takizawa.
\newblock Fault-tolerant flocking for a group of autonomous mobile robots.
\newblock {\em J. Syst. Softw.}, 84(1):29--36, 2011.
\newblock \href {https://doi.org/10.1016/j.jss.2010.08.026} {\path{doi:10.1016/j.jss.2010.08.026}}.

\end{thebibliography}

% % Short paper with appendix
% \mypart
% \setprefix{}
% \excludecomment{appendixonly}
% \includecomment{mainonly}
% \input{content.tex}

% \newpage
% \bibliographystyle{plainurl}
% \bibliography{main}

% \newpage
% \mypart
% \setprefix{apx:}
% \section*{Appendix}
% \includecomment{appendixonly}
% \excludecomment{mainonly}
% \input{content.tex}

\end{document}